\documentclass[11pt,a4paper]{article}

\usepackage[DIV=10]{typearea} 

\usepackage[utf8]{inputenc}
\usepackage[T1]{fontenc}

\usepackage{microtype}

\usepackage{amsmath}
\usepackage{amssymb}
\usepackage{amsthm}

\usepackage{libertine}
\usepackage[libertine]{newtxmath} 
\usepackage{MnSymbol}

\usepackage{tikz}

\usepackage{xcolor}
\usepackage{xspace}

\usepackage{comment}

\usepackage[
	style=alphabetic,
	backref=true,
	doi=false,
	url=false,
	maxcitenames=3,
	mincitenames=3,
	maxbibnames=10,
	minbibnames=10,
	backend=biber,
	sortlocale=en_US
]{biblatex}

\usepackage[ocgcolorlinks]{hyperref} 

\allowdisplaybreaks[1]

\renewcommand{\epsilon}{\varepsilon}

\protected\def\mathbb#1{\text{\usefont{U}{msb}{m}{n}#1}} 

\makeatletter
\g@addto@macro\bfseries{\boldmath}
\makeatother

\makeatletter
\g@addto@macro\mdseries{\unboldmath}
\g@addto@macro\normalfont{\unboldmath}
\g@addto@macro\rmfamily{\unboldmath}
\g@addto@macro\upshape{\unboldmath}
\makeatother

\DeclareCiteCommand{\citem}
    {}
    {\mkbibbrackets{\bibhyperref{\usebibmacro{postnote}}}}
    {\multicitedelim}
    {}

\renewcommand*{\multicitedelim}{\addcomma\space}

\AtEveryCitekey{\clearfield{note}}

\newcommand{\myhref}[1]{%
  \iffieldundef{doi}
    {\iffieldundef{url}
       {#1}
       {\href{\strfield{url}}{#1}}}
    {\href{http://dx.doi.org/\strfield{doi}}{#1}}%
}

\DeclareFieldFormat{title}{\myhref{\mkbibemph{#1}}}
\DeclareFieldFormat
  [article,inbook,incollection,inproceedings,patent,thesis,unpublished]
  {title}{\myhref{\mkbibquote{#1\isdot}}}

\makeatletter
\AtBeginDocument{%
    \newlength{\temp@x}%
    \newlength{\temp@y}%
    \newlength{\temp@w}%
    \newlength{\temp@h}%
    \def\my@coords#1#2#3#4{%
      \setlength{\temp@x}{#1}%
      \setlength{\temp@y}{#2}%
      \setlength{\temp@w}{#3}%
      \setlength{\temp@h}{#4}%
      \adjustlengths{}%
      \my@pdfliteral{\strip@pt\temp@x\space\strip@pt\temp@y\space\strip@pt\temp@w\space\strip@pt\temp@h\space re}}%
    \ifpdf
      \typeout{In PDF mode}%
      \def\my@pdfliteral#1{\pdfliteral page{#1}}
      \def\adjustlengths{}%
    \fi
    \ifxetex
      \def\my@pdfliteral #1{}
      \def\adjustlengths{\setlength{\temp@h}{-\temp@h}\addtolength{\temp@y}{1in}\addtolength{\temp@x}{-1in}}%
    \fi%
    \def\Hy@colorlink#1{%
      \begingroup
        \ifHy@ocgcolorlinks
          \def\Hy@ocgcolor{#1}%
          \my@pdfliteral{q}%
          \my@pdfliteral{7 Tr}
        \else
          \HyColor@UseColor#1%
        \fi
    }%
    \def\Hy@endcolorlink{%
      \ifHy@ocgcolorlinks%
        \my@pdfliteral{/OC/OCPrint BDC}%
        \my@coords{0pt}{0pt}{\pdfpagewidth}{\pdfpageheight}%
        \my@pdfliteral{F}
        \my@pdfliteral{EMC/OC/OCView BDC}%
        \begingroup%
          \expandafter\HyColor@UseColor\Hy@ocgcolor%
          \my@coords{0pt}{0pt}{\pdfpagewidth}{\pdfpageheight}%
          \my@pdfliteral{F}
        \endgroup%
        \my@pdfliteral{EMC}%
        \my@pdfliteral{0 Tr}
        \my@pdfliteral{Q}%
      \fi
      \endgroup
    }%
}
\makeatother

\colorlet{DarkRed}{red!50!black}
\colorlet{DarkGreen}{green!50!black}
\colorlet{DarkBlue}{blue!50!black}

\hypersetup{
	linkcolor = DarkRed,
	citecolor = DarkGreen,
	urlcolor = DarkBlue,
	bookmarksnumbered = true,
	linktocpage = true
}

\newtheorem{theorem}{Theorem}[section]
\newtheorem{lemma}[theorem]{Lemma}
\newtheorem{corollary}[theorem]{Corollary}

\newtheorem{definition}[theorem]{Definition}

\newtheorem{observation}[theorem]{Observation}

\newcommand{\dist}{\ensuremath{d}}

\title{Improved Algorithms for Computing the Cycle of Minimum Cost-to-Time Ratio in Directed Graphs\thanks{Accepted to the 44th International Colloquium on Automata, Languages, and Programming (ICALP 2017).}}

\author{Karl~Bringmann\thanks{Max Planck Institute for Informatics, Saarland Informatics Campus, Germany. Work partially done while visiting Aarhus University.}
	\and Thomas Dueholm Hansen\thanks{Aarhus University, Denmark. Supported by the Carlsberg Foundation, grant no. CF14-0617.}
	\and Sebastian~Krinninger\thanks{University of Vienna, Faculty of Computer Science, Austria. Work partially done while at Max Planck Institute for Informatics and while visiting Aarhus University.}
}
\date{}

\hypersetup{
	pdftitle = {Improved Algorithms for Computing the Cycle of Minimum Cost-to-Time Ratio in Directed Graphs},
	pdfauthor = {Karl Bringmann, Thomas Dueholm Hansen, Sebastian Krinninger}
}

\begin{document}

\maketitle
\begin{abstract}
We study the problem of finding the cycle of minimum cost-to-time ratio in a directed graph with $ n $ nodes and $ m $ edges.
This problem has a long history in combinatorial optimization and has recently seen interesting applications in the context of quantitative verification.
We focus on strongly polynomial algorithms to cover the use-case where the weights are relatively large compared to the size of the graph.
Our main result is an algorithm with running time $ \tilde O (m^{3/4} n^{3/2}) $, which gives the first improvement over Megiddo's $ \tilde O (n^3) $ algorithm \citem[JACM'83]{Megiddo83} for sparse graphs.\footnote{We use the notation $ \tilde O(\cdot) $ to hide factors that are polylogarithmic in $ n $.}
We further demonstrate how to obtain both an algorithm with running time $ n^3 / 2^{\Omega{(\sqrt{\log n})}} $ on general graphs and an algorithm with running time $ \tilde O (n) $ on constant treewidth graphs.
To obtain our main result, we develop a parallel algorithm for negative cycle detection and single-source shortest paths that might be of independent interest.

\end{abstract}

\section{Introduction}

We revisit the problem of computing the cycle of minimum cost-to-time ratio (short: minimum ratio cycle) of a directed graph in which every edge has a cost and a transit time.
The problem has a long history in combinatorial optimization and has recently become relevant to the computer-aided verification community in the context of quantitative verification and synthesis of reactive systems~\cite{ChakrabartiAHS03,ChatterjeeDH10,DrosteKV09,BloemCHJ09,CernyCHRS11,BloemCGHHJKK14,ChatterjeeIP15}.
The shift from quantitative to qualitative properties is motivated by the necessity of taking into account the resource consumption of systems (such as embedded systems) and not just their correctness.
For algorithmic purposes, these systems are usually modeled as directed graphs where vertices correspond to states of the system and edges correspond to transitions between states.
Weights on the edges model the resource consumption of transitions.
In our case, we allow two types of resources (called cost and transit time) and are interested in optimizing the ratio between the two quantities.
By giving improved algorithms for finding the minimum ratio cycle we contribute to the algorithmic progress that is needed to make the ideas of quantitative verification and synthesis applicable.

From a purely theoretic point of view, the minimum ratio problem is an interesting generalization of the minimum mean cycle problem.\footnote{In the minimum cycle mean problem we assume the transit time of each edge is $ 1 $.}
A natural question is whether the running time for the more general problem can match the running time of computing the minimum cycle mean (modulo lower order terms).
In terms of weakly polynomial algorithms, the answer to this question is yes, since a binary search over all possible values reduces the problem to negative cycle detection.
In terms of strongly polynomial algorithms, with running time independent of the encoding size of the edge weights, the fastest algorithm for the minimum ratio cycle problem is due to Megiddo~\cite{Megiddo83} and runs in time $ \tilde O (n^3) $, whereas the minimum mean cycle can be computed in $ O (m n) $ time with Karp's algorithm~\cite{Karp78}.
This has left an undesirable gap in the case of sparse graphs for more than three decades.

\subparagraph{Our results.}
We improve upon this situation by giving a strongly polynomial time algorithm for computing the minimum ratio cycle in time $ O (m^{3/4} n^{3/2} \log^{2} n) $ (Theorem~\ref{thm:deterministic min ratio algorithm} in Section~\ref{sec:deterministic}).
We obtain this result by designing a suitable parallel negative cycle detection algorithm and combining it with Megiddo's parametric search technique~\cite{Megiddo83}.
We first present a slightly simpler randomized version of our algorithm with one-sided error and the same running time (Theorem~\ref{thm:randomized min ratio algorithm} in Section~\ref{sec:randomized}).

As a side result, we develop a new parallel algorithm for negative cycle detection and single-source shortest paths (SSSP) that we use as a subroutine in the minimum ratio cycle algorithm.
This new algorithm has work $ \tilde O (mn + n^3 h^{-3}) $ and depth $ O(h) $ for any $ \log{n} \leq h \leq n $.
Our algorithm uses techniques from the parallel transitive closure algorithm of Ullman and Yannakakis~\cite{UllmanY91} (in particular as reviewed in~\cite{KleinS97}) and our contribution lies in extending these techniques to directed graphs with positive \emph{and negative} edge weights.
In particular, we partially answer an open question by Shi and Spencer~\cite{Spencer97} who previously gave similar trade-offs for single-source shortest paths in \emph{undirected} graphs with positive edge weights.
We further demonstrate how the parametric search technique can be applied to obtain minimum ratio cycle algorithms with running time $ \tilde O (n) $ on constant treewidth graphs (Corollary~\ref{cor:min ratio constant treewidth} in Section~\ref{sec:treewidth}).
Our algorithms do not use fast matrix multiplication. We finally show that if fast matrix multiplication is allowed then slight further improvements are possible, specifically we present an 
$ n^3 / 2^{\Omega{(\sqrt{\log n})}} $ time algorithm on general graphs (Corollary~\ref{cor:min ratio dense graphs} in Section~\ref{sec:dense graphs}).

\subparagraph{Prior Work.}

The minimum ratio problem was introduced to combinatorial optimization in the 1960s by Dantzig, Blattner, and Rao~\cite{DantzigBR67} and Lawler~\cite{Lawler67}.
The existing algorithms can be classified according to their running time bounds as follows: strongly polynomial algorithms, weakly polynomial algorithms, and pseudopolynomial algorithms.
In terms of strongly polynomial algorithms for finding the minimum ratio cycle we are aware of the following two results:
\begin{itemize}
\item $ O (n^3 \log{n} + m n \log^2{n}) $ time using Megiddo's second algorithm~\cite{Megiddo83} together with Cole's technique to reduce a factor of $ \log{\log{n}} $~\cite{Cole87},
\item $ O (m n^2) $ time using Burn's primal-dual algorithm~\cite{Burns91}.
\end{itemize}
For the class of weakly polynomial algorithms, the best algorithm is to follow Lawler's binary search approach~\cite{Lawler67,Lawler76}, which solves the problem by performing $ O (\log{(n W)}) $ calls to a negative cycle detection algorithm.
Here $ W = O (C T) $ if the costs are given as integers from $ 1 $ to $ C $ and the transit times are given as integers from $ 1 $ to $ T $.
Using an idea for efficient search of rationals~\cite{Papadimitriou79}, a somewhat more refined analysis by Chatterjee et al.~\cite{ChatterjeeIP15} reveals that it suffices to call the negative cycle detection algorithm $ O (\log(|a \cdot b |)) $ times when the value of the minimum ratio cycle is $ \tfrac{a}{b} $.
Since the initial publication of Lawler's idea, the state of the art in negative cycle detection algorithms has become more diverse.
Each of the following five algorithms gives the best known running time for some range of parameters (and the running times have to be multiplied by the factor $\log{(n W)} $ or $ \log(|a \cdot b |) $ to obtain an algorithm for the minimum ratio problem):
\begin{itemize}
\item $ O (mn) $ time using a variant of the Bellman-Ford algorithm~\cite{Ford56,Bellman58,Moore59},
\item $ n^3 / 2^{\Omega(\sqrt{\log n})} $ time using a recent all-pairs shortest paths (APSP) algorithm by Williams~\cite{Williams14,ChanW16},
\item $ \tilde O (n^{\omega} W) $ time using fast matrix multiplication~\cite{Sankowski05,YusterZ05}, where $ 2 \leq \omega < 2.3728639 $ is the matrix multiplication coefficient~\cite{Gall14a},
\item $ O (m \sqrt{n} \log W) $ time using Goldberg's scaling algorithm~\cite{Goldberg95},
\item $ \tilde O (m^{10/7} \log W) $ time using the interior point method based algorithm of Cohen et al.~\cite{CohenMSV17}
\end{itemize}
The third group of minimum ratio cycle algorithms has a pseudopolynomial running time bound.
After some initial progress~\cite{IshiiLP91,GerezHH92,ItoP95}, the state of the art is an algorithm by Hartmann and Orlin~\cite{HartmannO93} that has a running time of $ O (m n T) $.\footnote{Note that the more fine-grained analysis of Hartmann and Orlin actually gives a running time of $ O (m (\sum_{u \in V} \max_{e = (u, v)} t (e))) $.}
Other algorithmic approaches, without claiming any running time bounds superior to those reported above, were given by Fox~\cite{Fox69}, v.\ Golitschek~\cite{Golitschek82}, and Dasdan, Irani, and Gupta~\cite{DasdanIG99}.

Recently, the minimum ratio problem has been studied specifically for the special case of constant treewidth graphs by Chatterjee, Ibsen-Jensen, and Pavlogiannis~\cite{ChatterjeeIP15}.
The state of the art for negative cycle detection on constant treewidth graphs is an algorithm by Chaudhuri and Zaroliagis with running time $ O (n) $~\cite{ChaudhuriZ00}, which by Lawler's binary search approach implies an algorithm for the minimum ratio problem with running time $ O (n \log{(n W)}) $.
Chatterjee et al.~\cite{ChatterjeeIP15} report a running time of $ O (n \log(|a \cdot b |)) $ based on the more refined binary search mentioned above and additionally give an algorithm that uses $ O (\log{n}) $ space (and hence polynomial time).

As a subroutine in our minimum ratio cycle algorithm, we use a new parallel algorithm for negative cycle detection and single-source shortest paths.
The parallel SSSP problem has received considerable attention in the literature~\cite{Spencer97,KleinS97,Cohen97,BrodalTZ98,ShiS99,Cohen00,MeyerS03,MillerPVX15,Blelloch0ST16}, but we are not aware of any parallel SSSP algorithm that works in the presence of negative edge weights (and thus solves the negative cycle detection problem).
To the best of our knowledge, the only strongly polynomial bounds reported in the literature are as follows:
For weighted, directed graphs with non-negative edge weights, Broda, Träff, and Zaroliagis~\cite{BrodalTZ98} give an implementation of Dijkstra's algorithm with $ O (m \log{n}) $ work and $ O (n) $ depth.
For weighted, undirected graphs with positive edge weights, Shi and Spencer~\cite{ShiS99} gave (1) an algorithm with $ O (n^3 t^{-2} \log{n} \log{(n t^{-1})} + m \log{n}) $ work and $ O (t \log{n}) $ depth and (2) an algorithm with $ O ((n^3 t^{-3} + m n t^{-1}) \log{n}) $ work and $ O (t \log{n}) $ depth, for any $ \log{n} \leq t \leq n $.

\section{Preliminaries}

In the following, we review some of the tools that we use in designing our algorithm.

\subsection{Parametric Search}

We first explain the parametric search technique as outlined in~\cite{AgarwalST94}.
Assume we are given a property $ \mathcal{P} $ of real numbers that is \emph{monotone} in the following way: if $ \mathcal{P} (\lambda) $ holds, then also $ \mathcal{P} (\lambda') $ holds for all $ \lambda' < \lambda $.
Our goal is to find $ \lambda^* $, the \emph{maximum} $ \lambda $ such that $ \mathcal{P} (\lambda) $ holds.
In this paper for example, we will associate with each $ \lambda $ a weighted graph $ G_\lambda $ and $ \mathcal{P} $ is the property that $ G_\lambda $ has no negative cycle. 
Assume further that we are given an algorithm~$ \mathcal{A} $ for deciding, for a given $ \lambda $, whether $ \mathcal{P} (\lambda) $ holds.
If $ \lambda $ were known to only assume integer or rational values, we could solve this problem by performing binary search with $ O (\log{W}) $ calls to the decision algorithm, where $ W $ is the number of possible values for $ \lambda $.
However, this solution has the drawback of not yielding a strongly polynomial algorithm.

In parametric search we run the decision algorithm `generically' at the maximum $ \lambda^{*} $.
As the algorithm does not know $ \lambda^{*} $, we need to take care of its control flow ourselves and any time the algorithm performs a comparison we have to `manually' evaluate the comparison on behalf of the algorithm.
If each comparison takes the form of testing the sign of an associated low-degree polynomial $ p (\lambda) $, this can be done as follows.
We first determine all roots of $ p (\lambda) $ and check if $ \mathcal{P} (\lambda') $ holds for each such root $ \lambda' $ using another instance of the decision algorithm~$ \mathcal{A} $.
This gives us an interval between successive roots containing $ \lambda^* $ and we can thus resolve the comparison.
With every comparison we make, the interval containing $ \lambda^* $ shrinks and at the end of this process we can output a single candidate.
If the decision algorithm~$ \mathcal{A} $ has a running time of $ T $, then the overall algorithm for computing $ \lambda^* $ has a running time of $ O (T^2) $.

A more sophisticated use of the technique is possible, if in addition to a sequential decision algorithm~$ \mathcal{A}_\text{s} $ we have an efficient parallel decision algorithm~$ \mathcal{A}_\text{p} $.
The parallel algorithm performs its computations simultaneously on $ P_\text{p} $ processors.
The number of parallel computation steps until the last processor is finished is called the \emph{depth} $ D_\text{p} $ of the algorithm, and the number of operations performed by all processors in total is called the \emph{work} $ W_\text{p} $ of the algorithm.\footnote{To be precise, we use an abstract model of parallel computation as formalized in~\cite{FriedrichsL16} to avoid distraction by details such as read or write collisions typical to PRAM models.}
For parametric search, we actually only need parallelism w.r.t.\ comparisons involving the input values. 
We denote by the \emph{comparison depth} of $ \mathcal{A}_\text{p} $ the number of parallel comparisons (involving input values) until the last processor is finished.

We proceed similar to before:
We run $ \mathcal{A}_\text{p} $ `generically' at the maximum $ \lambda^{*} $ and (conceptually) distribute the work among $ P_\text{p} $ processors.
Now in each parallel step, we might have to resolve up to $ P_\text{p} $ comparisons.
We first determine all roots of the polynomials associated to these comparisons.
We then perform a binary search among these roots to determine the interval of successive roots containing $ \lambda^* $ and repeat this process of resolving comparisons at every parallel step to eventually find out the value of $ \lambda^{*} $.
If the sequential decision algorithm~$ \mathcal{A}_\text{s} $ has a running time of $ T_\text{s} $ and the parallel decision algorithm runs on $ P_\text{p} $ processors in $ D_\text{p} $ parallel steps, then the overall algorithm for computing $ \lambda^* $ has a running time of $ O (P_\text{p} D_\text{p} + D_\text{p} T_\text{s} \log{P_\text{p}}) $.
Formally, the guarantees of the technique we just described can be summarized as follows.

\begin{theorem}[\cite{AgarwalST94,Megiddo83}]\label{thm:parametric search}
Let $ \mathcal{P} $ be a property of real numbers such that if $ \mathcal{P} (\lambda) $ holds, then also $ \mathcal{P} (\lambda') $ holds for all $ \lambda' < \lambda $ and let $ \mathcal{A}_\text{p} $ and $ \mathcal{A}_\text{s} $ be algorithms deciding for a given $ \lambda $ whether $ \mathcal{P} (\lambda) $ holds such that
\begin{itemize}
\item the control flow of $ \mathcal{A}_\text{p} $ is only governed by comparisons that test the sign of an associated polynomial in $ \lambda $ of constant degree,
\item $ \mathcal{A}_\text{p} $ is a parallel algorithm with work $ W_\text{p} $ and comparison depth $ D_\text{p} $, and
\item $ \mathcal{A}_\text{s} $ is a sequential algorithm with running time $ T_\text{s} $.
\end{itemize}
Then there is a (sequential) algorithm for finding the maximum value $ \lambda $ such that $ \mathcal{P} (\lambda) $ holds with running time $ O (W_\text{p} + D_\text{p} T_\text{s} \log{W_\text{p}}) $.
\end{theorem}

Note that $ \mathcal{A}_\text{p} $ and $ \mathcal{A}_\text{s} $ need not necessarily be different algorithms.
In most cases however, the fastest sequential algorithm might be the better choice for minimizing running time.

\subsection{Characterization of Minimum Ratio Cycle}

We consider a directed graph $ G = (V, E, c, t) $, in which every edge $ e = (u, v) $ has a cost $ c (e) $ and a transit time $ t (e) $.
We want to find the cycle $ C $ that minimizes the cost-to-time ratio $ \sum_{e \in C} c (e) / \sum_{e \in C} t (e) $.

For any real $ \lambda $ define the graph $ G_\lambda = (V, E, w_\lambda) $ as the modification of $ G $ with weight $ w_\lambda (e) = c (e) - \lambda t (e) $ for every edge $ e \in E $.
The following structural lemma is the foundation of many algorithmic approaches towards the problem.

\begin{lemma}[\cite{DantzigBR67,Lawler76}]\label{lem:characterization of minimum ratio cycle}
Let $ \lambda^* $ be the value of the minimum ratio cycle of $ G $.
\begin{itemize}
\item For $ \lambda > \lambda^* $, the value of the minimum weight cycle in $ G_\lambda $ is $ < 0 $.
\item The value of the the minimum weight cycle in $ G_{\lambda^*} $ is $ 0 $. Each minimum weight cycle in $ G_{\lambda^*} $ is a minimum ratio cycle in $ G $ and vice versa.
\item For $ \lambda < \lambda^* $, the value of the minimum weight cycle in $ G_\lambda $ is $ > 0 $.
\end{itemize}
\end{lemma}

The obvious algorithmic idea now is to find the right value of $ \lambda $ with a suitable search strategy and reduce the problem to a series of negative cycle detection instances.

\subsection{Characterization of Negative Cycle}

\begin{definition}
A \emph{potential function} $ p\colon V \to \mathbb{R} $ assigns a value to each vertex of a weighted directed graph $ G = (V, E, w) $.
We call a potential function $ p $ \emph{valid} if for every edge $ e = (u, v) \in E $, the condition $ p(u) + w (e) \geq p (v) $ holds.
\end{definition}

The following two lemmas outline an approach for negative cycle detection.

\begin{lemma}[\cite{EdmondsK72}]\label{lem:negative cycle characterization}
A weighted directed graph contains a negative cycle if and only if it has no valid potential function.
\end{lemma}

\begin{lemma}[\cite{Johnson77}]\label{lem:distances give valid potential}
Let $ G = (V, E, w) $ be a weighted directed graph and let $ G' = (V', E', w') $ be the supergraph of $ G $ consisting of the vertices $ V' = V \cup \{ s' \} $ (i.e.\ with an additional super-source~$s'$), the edges $ E' = E \cup \{ s' \} \times V $ and the weight function $ w' $ given by $ w' (s', v) = 0 $ for every vertex $ v \in V $ and $ w' (u, v) = w (u, v) $ for all pairs of vertices $ u, v \in V $.
If $ G $ does not contain a negative cycle, then the potential function $ p $ defined by $ p (v) = \dist_{G'} (s', v) $ for every vertex $ v \in V $ is valid for $G$.
\end{lemma}

Thus, an obvious strategy for negative cycle detection is to design a single-source shortest paths algorithm that is correct whenever the graph contains no negative cycle.
If the graph contains no negative cycle, then the distances computed by the algorithm can be verified to be a valid potential.
If the graph does contain a negative cycle, then the distances computed by the algorithm will not be a valid potential (because a valid potential does not exist) and we can verify that the potential is not valid.

\subsection{Computing Shortest Paths in Parallel}

In our algorithm we use two building blocks for computing shortest paths in the presence of negative edge weights in parallel.
The first such building block was also used by Megiddo~\cite{Megiddo83}.

\begin{observation}\label{lem:min plus}
By repeated squaring of the min-plus matrix product, all-pairs shortest paths in a directed graph with real edge weights can be computed using work $ O (n^3 \log{n}) $ and depth $ O (\log{n}) $.
\end{observation}

The second building block is a subroutine for computing the following restricted version of shortest paths.

\begin{definition}
The \emph{shortest $h$-hop path} from a vertex $ s $ to a vertex $ t $ is the path of minimum weight among all paths from $ s $ to $ t $ with at most $ h $ edges.
\end{definition}

Note that a shortest $h$-hop path from $ s $ to $ t $ does not exist, if all paths from $ s $ to $ t $ use more than $ h $ edges. 
Furthermore, if there is a shortest path from $ s $ to $ t $ with at most $ h $ edges, then the $ h $-hop shortest path from $ s $ to $ t $ is a shortest path as well.
Shortest $h$-hop paths can be computed by running $ h $ iterations of the Bellman-Ford algorithm~\cite{Ford56,Bellman58,Moore59}.\footnote{The first explicit use of the Bellman-Ford algorithm to compute shortest $h$-hop paths that we are aware of is in Thorup's dynamic APSP algorithm~\cite{Thorup05}.}
Similar to shortest paths, shortest $h$-hop paths need not be unique.
We can enforce uniqueness by putting some arbitrary but fixed order on the vertices of the graph and sorting paths according to the induced lexicographic order on the sequence of vertices of the paths.
Note that the Bellman-Ford algorithm can easily be adapted to optimizing lexicographically as its second objective.

\begin{observation}\label{lem:Bellman-Ford}
By performing $ h $ iterations of the Bellman-Ford algorithm, the lexicographically smallest shortest $h$-hop path from a designated source vertex $ s $ to each other vertex in a directed graph with real edge weights can be computed using work $ O (m h) $ and depth $ O (h) $.
\end{observation}

We denote by $ \pi (s, t) $ the lexicographically smallest shortest path from $ s $ to $ t $ and by $ \pi^h (s, t) $ the lexicographically smallest shortest $h$-hop path from $ s $ to $ t $.
We denote by $ V (\pi^h (s, t)) $ and $ E (\pi^h (s, t)) $ the set of nodes and edges of $ \pi^h (s, t) $, respectively.

\subsection{Approximate Hitting Sets}

\begin{definition}
Given a collection of sets $ \mathcal{S} \subseteq 2^{U} $ over a universe $ U $, a hitting set is a set $ T \subseteq H $ that has non-empty intersection with every set of $ \mathcal{S} $ (i.e., $ S \cap T \neq \emptyset $ for every $ S \in \mathcal{S} $).
\end{definition}

Computing a hitting set of minimum size is an \textsf{NP}-hard problem.
For our purpose however, rough approximations are good enough.
The first method to get a sufficiently small hitting set uses a simple randomized sampling idea and was introduced to the design of graph algorithms by Ullman and Yannakakis~\cite{UllmanY91}.
We use the following formulation.

\begin{lemma}\label{lem:randomized hitting set}
Let $ c \geq 1 $, let $ U $ be a set of size $ s $ and let $ \mathcal{S} = \{ S_1, S_2, \ldots, S_k \} $ be a collection of sets over the universe $ U $ of size at least $ q $.
Let~$ T $ be a subset of $ U $ that was obtained by choosing each element of $ U $ independently with probability $ p = \min (x / q, 1) $ where $ x = c \ln{(k s)} + 1 $.
Then, with high probability (whp), i.e., probability at least $ 1 - 1/s^c $, the following two properties hold:
\begin{enumerate}
\item For every $ 1 \leq i \leq k $, the set $ S_i $ contains an element of $ T $, i.e., $ S_i \cap T \neq \emptyset $.
\item $ |T| \leq 3 x s / q = O (c s \log{(k s)} / q) $.
\end{enumerate}
\end{lemma}

The second method is to use a heuristic to compute an approximately minimum hitting set.
In the sequential model, a simple greedy algorithm computes an $ O(\log{n}) $-approximation~\cite{Johnson74a,AusielloDP80}. We use the following formulation.

\begin{lemma}\label{lem:greedy hitting set}
Let $ U $ be a set of size $ s $ and let $ \mathcal{S} = \{ S_1, S_2, \ldots, S_k \} $ be a collection of sets over the universe $ U $ of size at least $ q $.
Consider the simple greedy algorithm that picks an element $u$ in $U$ that is contained in the largest number of sets in $\mathcal{S}$ and then removes all sets containing $u$ from $\mathcal{S}$, repeating this step until $\mathcal{S} = \emptyset$.
Then the set $T$ of elements picked by this algorithm satisfies:
\begin{enumerate}
\item For every $ 1 \leq i \leq k $, the set $ S_i $ contains an element of $ T $, i.e., $ S_i \cap T \neq \emptyset $.
\item $ |T| \leq O (s \log{(k)} / q) $.
\end{enumerate}
\end{lemma}
\begin{proof}
  We follow the standard proof of the approximation ratio $O(\log n)$ for the greedy set cover heuristic. 
  The first statement is immediate, since we only remove sets when they are hit by the picked element.
  Since each of the $k$ sets contains at least $q$ elements, on average each element in $U$ is contained in at least $k q / s$ sets. Thus, the element $u$ picked by the greedy algorithm is contained in at least $k q / s$ sets. The remaining number of sets is thus at most $k - k q / s = k(1-q/s)$. Note that the remaining sets still have size at least $q$, since they do not contain the picked element $u$. Inductively, we thus obtain that after $i$ iterations the number of remaining sets is at most $k(1-q/s)^i$, so after $O(\log(k) \cdot s/q)$ iterations the number of remaining sets is less than 1 and the process stops.
\end{proof}

The above greedy algorithm is however inherently sequential and thus researchers have studied more sophisticated algorithms for the parallel model.
The state of the art in terms of deterministic algorithms is an algorithm by Berger et al.~\cite{BergerRS94}\footnote{Berger et al.\ actually give an approximation algorithm for the following slightly more general problem: Given a hypergraph $ H = (V, E) $ and a cost function $ c\colon V \to \mathbb{R} $ on the vertices, find a minimum cost subset $ R \subseteq V $ that covers $ H $, i.e., an $ R $ that minimizes $ c (R) = \sum_{v \in R} c (v) $ subject to the constraint $ e \cap R \neq \emptyset $ for all $ e \in E $.}.

\begin{theorem}[\cite{BergerRS94}]\label{thm:deterministic hitting set algorithm}
Let $ \mathcal{S} = \{ S_1, S_2, \ldots, S_k \} $ be a collection of sets over the universe $ U $, let $ n = |U| $ and $ m = \sum_{1 \leq i \leq k} |S_i| $.
For $ 0 < \epsilon < 1 $, there is an algorithm with work $ O ((m + n) \epsilon^{-6} \log^4{n} \log{m} \log^6{(nm)}) $ and depth $ O (\epsilon^{-6} \log^4{n} \log{m} \log^6{(nm)}) $ that produces a hitting set of $ \mathcal{S} $ of size at most $ (1 + \epsilon) (1 + \ln{\Delta}) \cdot \mathit{OPT} $, where $ \Delta $ is the maximum number of occurrences of any element of $ U $ in $ \mathcal{S} $ and $ \mathit{OPT} $ is the size of a minimum hitting set.
\end{theorem}

\section{Randomized Algorithm for General Graphs}\label{sec:randomized}

\subsection{A Parallel SSSP Algorithm}\label{sec:parallel SSSP}

In the following we design a parallel SSSP algorithm that can be used to check for negative cycles.
Formally, we will in this subsection prove the following statement.

\begin{theorem}\label{thm:SSSP algorithm}
There is an algorithm that, given a weighted directed graph G = (V, E, w) containing no negative cycles, computes the shortest paths from a designated source vertex $ s $ to all other vertices spending $ O (m n \log{n} + n^3 h^{-3} \log^4{n}) $ work with $ O (h + \log{n}) $ depth for any $ 1 \leq h \leq n $.
The algorithm is correct with high probability and all its comparisons are performed on sums of edge weights on both sides.
\end{theorem}

The algorithm proceeds in the following steps:
\begin{enumerate}
\item Let $ C \subseteq V $ be a set containing each vertex $ v $ independently with probability $ p = \min (3 c h^{-1} \ln{n}, 1) $ for a sufficiently large constant $ c $. \label{step:sample centers}
\item If $ | C | > 9 c n h^{-1} \ln{n} $, then terminate. \label{step:check size of hitting set}
\item For every vertex $ x \in C \cup \{ s \} $ and every vertex $ v \in V $, compute the shortest $h$-hop path from $ x $ to $ v $ in $ G $ and its weight $ \dist_G^h (x, v) $. \label{step:shortest h hop paths}
\item Construct the graph $ H = (C \cup \{ s \}, (C \cup \{ s \})^2, w_H) $ whose set of vertices is $ C \cup \{ s \} $, whose set of edges is $ (C \cup \{ s \})^2 $ and for every pair of vertices $ x, y \in C\cup \{ s \} $ the weight of the edge $ (x, y) $ is $ w_H (x, y) = \dist_G^h (x, y) $. \label{step:construct center graph}
\item For every vertex $ x \in C $, compute the shortest path from $ s $ to $ x $ in $ H $ and its weight $ \dist_H (s, x) $. \label{step:APSP on center graph}%
\item For every vertex $ t \in V $, set $ \delta (t) = \min_{x \in C \cup \{ s \}} (\dist_H (s, x) + \dist_G^h (x, t)) $. \label{step:compute output}
\end{enumerate}

\subsubsection{Correctness}

In order to prove the correctness of the algorithm, we first observe that as a direct consequence of Lemma~\ref{lem:randomized hitting set} the randomly selected vertices in $ C $ with high probability hit all lexicographically smallest shortest $ \lfloor h/2 \rfloor $-hop paths of the graph.

\begin{observation}\label{lem:hitting paths}
Consider the collection of sets
\begin{equation*}
\mathcal{S} = \{ V(\pi^{\lfloor h/2 \rfloor} (u, v)) \mid u, v \in V \text{ with } \dist_G^{\lfloor h/2 \rfloor} (u, v) < \infty \text{ and }  |E(\pi^{\lfloor h/2 \rfloor} (u, v))| = \lfloor h/2 \rfloor \}
\end{equation*}
containing the vertices of the lexicographically smallest shortest $\lfloor h/2 \rfloor$-hop paths with exactly $\lfloor h/2 \rfloor$ edges between all pairs of vertices.
Then, with high probability, $ C $ is a hitting set of $ \mathcal{S} $ of size at most $ 9 c n h^{-1} \ln{n} $.
\end{observation}

\begin{lemma}\label{lem:correctness randomized}
If $ G $ contains no negative cycle, then $ \delta (t) = \dist_G (s, t) $ for every vertex $ t \in V $ with high probability.
\end{lemma}

\begin{proof}
First note that the algorithm incorrectly terminates in Step~\ref{step:check size of hitting set} only with small probability.
We now need to show that, for every vertex $ t \in V $, $ \delta (t) := \min_{x \in C \cup \{s\}} (\dist_H (s, x) + \dist_G^h (x, t)) = \dist_G (s, t) $.
First observe that every edge in $ H $ corresponds to a path in $ G $ (of the same weight).
Thus, the value $ \delta (t) $ corresponds to some path in $ G $ from $ s $ to $ t $ (of the same weight) which implies that $ \dist_G (s, t) \leq \delta (t) $ (as no path can have weight less than the distance).

Now let $ \pi (s, t) $ be the lexicographically smallest shortest path from $ s $ to $ t $ in $ G $.
Subdivide $ \pi $ into consecutive subpaths $ \pi_1, \ldots, \pi_k $ such that $ \pi_i $ for $ 1 \leq i \leq k-1 $ has exactly $ \lfloor h/2 \rfloor $ edges, and $ \pi_k $ has at most $ \lfloor h/2 \rfloor $ edges.
Note that if $ \pi $ itself has at most $ \lfloor h/2 \rfloor $ edges, then $ k = 1 $.
Since every subpath of a lexicographically smallest shortest path is also a lexicographically smallest shortest path, the paths $ \pi_1, \ldots, \pi_k $ are lexicographically smallest shortest paths as well.
As the subpaths $ \pi_1, \ldots, \pi_{k-1} $ consist of exactly $ \lfloor h/2 \rfloor $ edges, each of them is contained in the collection of sets $ \mathcal{S} $ of Observation~\ref{lem:hitting paths}.
Therefore, each subpath $ \pi_i $, for $ 1 \leq i \leq k-1 $, contains a vertex $ x_i \in C $ with high probability.

Set $ x_0 = s $ and $ x_k = t $, and observe that for every $ 0 \leq i \leq k-1 $, the subpath of $ \pi (s, t) $ from $ x_i $ to $ x_{i+1} $ is a shortest path from $ x_i $ to $ x_{i+1} $ with at most $ h $ edges and thus $ \dist_G^h (x_i, x_{i+1}) = \dist_G (x_i, x_{i+1}) $.
We now get the following chain of inequalities:
\begin{align*}
\dist_G (s, t) = \sum_{0 \leq i \leq k-1} \dist_G (x_i, x_{i+1}) &= \sum_{0 \leq i \leq k-1} \dist_G^h (x_i, x_{i+1}) \\
 &= \Big( \sum_{0 \leq i \leq k-2} w_H (x_i, x_{i+1}) \Big) + \dist_G^h (x_{k-1}, t) \\
 &\geq \dist_H (x_0, x_{k-1}) + \dist_G^h (x_{k-1}, t) \\
 &= \dist_H (s, x_{k-1}) + \dist_G^h (x_{k-1}, t) \\
 &\geq \min_{x \in C \cup \{s\}} (\dist_H (s, x) + \dist_G^h (x, t)) = \delta(t) \, . \qedhere
\end{align*}
\end{proof}

Note that we have formally argued only that the algorithm correctly computes the \emph{distances} from $ s $.
It can easily be checked that the shortest paths can be obtained by replacing the edges of $ H $ with their corresponding paths in $ G $.

\subsubsection{Running Time}

\begin{lemma}
The algorithm above can be implemented with $ O (m n \log{n} + n^3 h^{-3} \log^4{n}) $ and $ O (h + \log{n}) $ depth such that all its comparisons are performed on sums of edge weights on both sides.
\end{lemma}

\begin{proof}
Clearly, in Steps~\ref{step:sample centers}--\ref{step:check size of hitting set}, the algorithm spends $ O (m + n) $ work with $ O (1) $ depth.
Step~\ref{step:shortest h hop paths} can be carried out by running $ h $ iterations of Bellman-Ford for every vertex $ x \in C $ in parallel (see Lemma~\ref{lem:Bellman-Ford}), thus spending $ O (|C| \cdot m h) $ work with $ O (h) $ depth.
Step~\ref{step:construct center graph} can be carried out by spending $ O (|C|^2) $ work with $ O(1) $ depth.
Step~\ref{step:APSP on center graph} can be carried out by running the min-plus matrix multiplication based APSP algorithm (see Lemma~\ref{lem:min plus}), thus spending $ O (|C|^3 \log{n}) $ work with $ O (\log{n}) $ depth.
The naive implementation of Step~\ref{step:compute output} spends $ O (n |C|) $ work with $ O (|C|) $ depth.
Using a bottom-up `tournament' approach where in each round we pair up all values and let the maximum value of each pair proceed to the next round, this can be improved to work $ O (n |C|) $ and depth $ O (\log{n}) $.

It follows that by carrying out the steps of the algorithm sequentially as explained above, the overall work is $ O (|C| \cdot m h + |C|^3 \log{n}) $ and the depth is $ O (h + \log{n}) $.
As the algorithm ensures that $ | C | \leq 9 c n h^{-1} \ln{n} $ for some constant $c$, the work is $ O (m n \log{n} + n^3 h^{-3} \log^4{n}) $ and the depth is $ O (h + \log{n}) $.
\end{proof}

\subsubsection{Extension to Negative Cycle Detection}

To check whether a weighted graph $ G = (V, E, w) $ contains a negative cycle, we first construct the graph $ G' $ (with an additional super-source $ s' $) as defined in Lemma~\ref{lem:distances give valid potential}.
We then run the SSSP algorithm of Theorem~\ref{thm:SSSP algorithm} from $ s' $ in $ G' $ and set $ p (v) = \dist_{G'} (s', t) $ for every vertex $ t \in V $.
We then check whether the function $ p $ defined in this way is a valid potential function for $ G $ testing for every edge $ e = (u, v) $ (in parallel) whether $ p (u) + w (u, v) \geq p (v) $.
If this is the case, then we output that $ G $ contains no negative cycle, otherwise we output that $ G $ contains a negative cycle.

\begin{corollary}\label{cor:parallel negative cycle detection}
There is a randomized algorithm that checks whether a given weighted directed graph contains a negative cycle with $ O (m n \log{n} + n^3 h^{-3} \log^4{n}) $ work and $ O (h + \log{n}) $ depth for any $ 1 \leq h \leq n $.
The algorithm is correct with high probability and all its comparisons are performed on sums of edge weights on both sides.
\end{corollary}

\begin{proof}
Constructing the graph $ G' $ and checking whether $ p $ is a valid potential can both be carried out with $ O (m + n) $ work and $ O(1) $ depth.
Thus, the overall work and depth bounds are asymptotically equal to the SSSP algorithm of Theorem~\ref{thm:SSSP algorithm}.

If $ G $ contains no negative cycle, then the SSSP algorithm correctly computes the distances from $ s' $ in $ G' $.
Thus, the potential $ p $ is valid by Lemma~\ref{lem:distances give valid potential} and our algorithm correctly outputs that there is no negative cycle.
If $ G $ contains a negative cycle, then it does not have any valid potential by Lemma~\ref{lem:negative cycle characterization}.
Thus, the potential $ p $ defined by the algorithm cannot be valid and the algorithm outputs correctly that $ G $ contains a negative cycle.
\end{proof}

\subsection{Finding the Minimum Ratio Cycle}

Using the negative cycle detection algorithm as a subroutine, we obtain an algorithm for computing a minimum ratio cycle in time $ \tilde O (n^{3/2} m^{3/4}) $.

\begin{theorem}\label{thm:randomized min ratio algorithm}
There is a randomized one-sided-error Monte Carlo algorithm for computing a minimum ratio cycle with running time $ O (n^{3/2} m^{3/4} \log^2{n}) $.
\end{theorem}

\begin{proof}
By Lemma~\ref{lem:characterization of minimum ratio cycle} we can compute the value of the minimum ratio cycle by finding the largest value of $ \lambda $ such that $ G_\lambda $ contains no negative-weight cycle.
We want to apply Theorem~\ref{thm:parametric search} to find this maximum $ \lambda^* $ by parametric search.
As the sequential negative cycle detection algorithm $ A_\text{s} $ we use Orlin's minimum weight cycle algorithm~\cite{Orlin17} with running time $ T(n, m) = O (m n) $.
The parallel negative cycle detection algorithm $ A_\text{p} $ of Corollary~\ref{cor:parallel negative cycle detection} has work $ W (n, m) = O (m n \log{n} + n^3 h^{-3} \log^4{n}) $ and depth $ D (n, m) = O (h + \log{n}) $, for any choice of $ 1 \leq h \leq n $.
Any comparison the latter algorithm performs is comparing sums of edge weights of the graph.
Since in $ G_\lambda $ edge weights are linear functions in $\lambda$, the control flow only depends on testing the sign of degree-1 polynomials in $ \lambda $.
Thus, Theorem~\ref{thm:parametric search} is applicable\footnote{Formally, Theorem~\ref{thm:parametric search} only applies to deterministic algorithms. However, only step~\ref{step:sample centers} of our parallel algorithm is randomized, but this step does not depend on $\lambda$. All remaining steps are deterministic. We can thus first perform steps~\ref{step:sample centers} and~\ref{step:check size of hitting set}, and invoke Theorem~\ref{thm:parametric search} only on the remaining algorithm. The output guarantee then holds with high probability.} and we arrive at a sequential algorithm for finding the value of the minimum ratio cycle with running time $ O (m n \log{n} (h + \log{n}) + n^3 h^{-3} \log^4{n}) $.
Finally, to output the minimum ratio cycle and not just its value, we run Orlin's algorithm for finding the minimum weight cycle in $ G_{\lambda^*} $, which takes time $ O (m n) $.
By setting $ h = n^{1/2} m^{-1/4} \log{n} $ the overall running time becomes $ O (n^{3/2} m^{3/4} \log^2{n}) $.
\end{proof}

\section{Deterministic Algorithm for General Graphs}\label{sec:deterministic}

We now present a deterministic variant of our minimum ratio cycle algorithm, with the same running time as the randomized algorithm up to logarithmic factors.

\subsection{Deterministic SSSP and Negative Cycle Detection}

We can derandomize our SSSP algorithm by combining a preprocessing step with the parallel hitting set approximation algorithm of~\cite{BergerRS94}.
Formally, we will prove the following statement.

\begin{theorem}\label{thm:deterministic SSSP algorithm}
There is a deterministic algorithm that, given a weighted directed graph containing no negative cycles, computes the shortest paths from a designated source vertex $ s $ to all other vertices spending $ O (m n \log^2{n} + n^3 h^{-3} \log^7{n} + n^2 h \log^{11}{n}) $ work with $ O (h + \log^{11}{n}) $ depth for any $ 1 \leq h \leq n $.
\end{theorem}

From this, using Lemmas~\ref{lem:negative cycle characterization} and~\ref{lem:distances give valid potential} analogously to the proof of Corollary~\ref{cor:parallel negative cycle detection}, we get the following corollary for negative cycle detection.
\begin{corollary}\label{cor:deterministic parallel negative cycle detection}
There is a deterministic algorithm that checks whether a given weighted directed graph contains a negative cycle with $ O (m n \log^2{n} + n^3 h^{-3} \log^7{n} + n^2 h \log^{11}{n}) $ work and $ O (h + \log^{11}{n}) $ depth for any $ 1 \leq h \leq n $.
\end{corollary}

Our deterministic SSSP algorithm does the following:
\begin{enumerate}
\item For all pairs of vertices $ u,v \in V $, compute the shortest $ \lfloor h/2 \rfloor $-hop path $ \pi^{\lfloor h/2 \rfloor} (u, v) $ from $ u $ to $ v $ in $ G $.\footnote{Note that in case there are multiple shortest $ \lfloor h/2 \rfloor $-hop paths from $ u $ to $ v $, any tie-breaking is fine for the algorithm and its analysis.} \label{step:compute shortest h hop paths}
\item Compute an $ O (\log{n}) $-approximate set cover $ C $ of the system of sets \begin{center} $\mathcal{S} = \{ V(\pi^{\lfloor h/2 \rfloor} (u, v)) \mid u, v \in V \text{ with } \dist_G^{\lfloor h/2 \rfloor} (u, v) < \infty \text{ and }  |E(\pi^{\lfloor h/2 \rfloor} (u, v))| = \lfloor h/2 \rfloor \} $. \end{center} \label{step:compute hitting set}
\item Proceed with steps~\ref{step:shortest h hop paths} to \ref{step:compute output} of the algorithm in Section~\ref{sec:parallel SSSP}.
\end{enumerate}

\subsubsection{Correctness}

Correctness is immediate:
In the previous proof of Lemma~\ref{lem:correctness randomized} we relied on the fact that $ C $ is a hitting set of $ \mathcal{S} $.
In the above algorithm, this property is guaranteed directly.

\subsubsection{Running Time}

Step~\ref{step:compute shortest h hop paths} can be carried out by running $ h $ iterations of the Bellman-Ford algorithm for every vertex $ v \in V $.
By Lemma~\ref{lem:Bellman-Ford} this uses $ O (m n h) $ work and $ O (h) $ depth.
We carry out Step~\ref{step:compute hitting set} by running the algorithm of Theorem~\ref{thm:deterministic hitting set algorithm} to compute an $ O (\log{n}) $-approximate hitting set of $ \mathcal{S} $ with work $ O (n^2 h \log^{11}{n}) $ and depth $ O (\log^{11}{n}) $.
Lemma~\ref{lem:randomized hitting set} gives a randomized process that computes a hitting set of $ \mathcal{S} $ of expected size $ O (n h^{-1} \log{n}) $.
By the probabilistic method, this implies that there exists a hitting set of size $ O (n h^{-1} \log{n}) $.
We can therefore use the algorithm of Theorem~\ref{thm:deterministic hitting set algorithm} to compute a hitting set $ \mathcal{S} $ of size $ O (n h^{-1} \log^2{n}) $.
The work is $ O (n^2 h \log^{11}{n}) $ and the depth is $ O (\log^{11}{n}) $.
Carrying out the remaining steps with a hitting set $ C $ of size $ O (n h^{-1} \log^2{n}) $ uses work $ O (m h |C| + |C|^3 \log{n}) = O (m n \log^2{n} + n^3 h^{-3} \log^7{n}) $ and depth $ O (h + \log{n}) $.
Thus, our overall SSSP algorithm has work $ O (m n \log^2{n} + n^3 h^{-3} \log^7{n} + n^2 h \log^{11}{n}) $ and depth $ O (h + \log^{11}{n}) $.

\subsection{Minimum Ratio Cycle}

We again obtain a minimum ratio cycle algorithm by applying parametric search (Theorem~\ref{thm:parametric search}).
We obtain the same running time bound as for the randomized algorithm. 

\begin{theorem}\label{thm:deterministic min ratio algorithm}
There is a deterministic algorithm for computing a minimum ratio cycle with running time $ O (n^{3/2} m^{3/4} \log^{2} n) $.
\end{theorem}

\begin{proof}[Proof sketch]
The proof is analogous to the proof of Theorem~\ref{thm:randomized min ratio algorithm}, with the only exception that we use the deterministic parallel negative cycle detection algorithm of Corollary~\ref{cor:deterministic parallel negative cycle detection}.
However, we do not necessarily need to run the algorithm of Theorem~\ref{thm:deterministic hitting set algorithm} to compute an approximate hitting set.
Instead we can also run the greedy set cover heuristic (Lemma~\ref{lem:greedy hitting set}) for this purpose.
The reason is that at this stage, the greedy heuristic does not need to perform any comparisons involving the edge weights of the input graph, which are the only operations that are costly in the parametric search technique.
This means that finding an approximate hitting set $C$ of size $O(n h^{-1} \log n)$ can be implemented with $ O (\sum_{S \in \mathcal{S}} |S|) = O (n^2 h) $ work and $ O (1) $ comparison depth.
Thus, we use a parallel negative cycle detection algorithm $ A_\text{p} $ which has work $ W (n, m) = O (m h |C| + |C|^3 \log{n} + n^2 h) = O (m n \log n + n^3 h^{-3} \log^4 n + n^2 h)$ and depth $ D (n, m) = O (h + \log n) $, for any choice of $ 1 \leq h \leq n $.
We thus obtain a sequential minimum ratio cycle algorithm with running time $ O (m n \log n + n^3 h^{-3} \log^4 n + n^2 h + m n \log{n} (h + \log n)) $, for any choice of $ 1 \leq h \leq n $.
Note that the summands $ m n \log n $ and $ n^2 h $ are both dominated by the last summand $ m n \log{n} (h + \log n) $. Setting $ h = n^{1/2} m^{-1/4} \log n $ to optimize the remaining summands, the running time becomes $ O (n^{3/2} m^{3/4} \log^2{n}) $.
\end{proof}

\section{Near-Linear Time Algorithm for Constant Treewidth Graphs}\label{sec:treewidth}

In the following we demonstrate how to obtain a nearly-linear time algorithm  (in the strongly polynomial sense) for graphs of constant treewidth.
We can use the following results of Chaudhuri and Zaroliagis~\cite{ChaudhuriZ00} who studied the shortest paths problem in graphs of constant treewidth.\footnote{The first result of Chaudhuri and Zaroliagis~\cite{ChaudhuriZ00} has recently been complemented with a space-time trade-off by Chatterjee, Ibsen-Jensen, and Pavlogiannis~\cite{ChatterjeeIP16} at the cost of polynomial preprocessing time that is too large for our purposes.}

\begin{theorem}[\cite{ChaudhuriZ00}]\label{thm:APSP treewidth sequential}
There is a deterministic algorithm that, given a weighted directed graph containing no negative cycles, computes a data structure that after $ O (n) $ preprocessing time can answer, for any pair of vertices, distance queries in time $ O (\alpha(n)) $, where $ \alpha (\cdot) $ is the inverse Ackermann function.
It can also report a corresponding shortest path in time $ O (\ell \alpha(n)) $, where $ \ell $ is the number of edges of the reported path.
\end{theorem}

\begin{theorem}[\cite{ChaudhuriZ98}]\label{thm:SSSP treewidth parallel}
There is a deterministic negative cycle algorithm for weighted directed graphs of constant treewidth with $ O (n) $ work and $ O (\log^2{n}) $ depth. 
\end{theorem}

We now apply the reduction of Theorem~\ref{thm:parametric search} to the algorithm of Theorem~\ref{thm:SSSP treewidth parallel} to find $ \lambda^* $, the value of the minimum ratio cycle, in time $ O (n \log^3{n}) $ (using $ T_\text{s} (n) = W_\text{p} (n) = O (n) $, and $ D_\text{p} (n) = O (\log^2{n}) $).
We then use the algorithm of Theorem~\ref{thm:APSP treewidth sequential} to find a minimum weight cycle in $ G_{\lambda^*} $ in time $ O (n \alpha(n)) $: Each edge $ e = (u, v) $ together with the shortest path from $ v $ to $ u $ (if it exists) defines a cycle and we need to find the one of minimum weight by asking the corresponding distance queries.
For the edge $ e = (u, v) $ defining the minimum weight cycle we query for the corresponding shortest path from $ v $ to $ u $.
This takes time $ O (n) $ as a graph of constant treewidth has $ O (n) $ edges.
We thus arrive at the following guarantees of the overall algorithm.

\begin{corollary}\label{cor:min ratio constant treewidth}
There is a deterministic algorithm that computes the minimum ratio cycle in a directed graph of constant treewidth in time $ O (n \log^3{n}) $.
\end{corollary}

\section{Slightly Faster Algorithm for Dense Graphs}\label{sec:dense graphs}

All our previous algorithm do not make use of fast matrix multiplication. We now show that if we allow fast matrix multiplication, despite the hidden constant factors being galactic, then slight further improvements are possible.
Specifically, we sketch how the running time of $ n^3 / 2^{\Omega{(\sqrt{\log n})}} $ of Williams's recent APSP algorithm~\cite{Williams14} (with a deterministic version by Chan and Williams~\cite{ChanW16}) can be salvaged for the minimum ratio problem.
In particular, we explain why Williams' algorithm for min-plus matrix multiplication parallelizes well enough.

\begin{theorem}\label{thm:parallel min plus}
There is a deterministic algorithm that checks whether a given weighted directed graph contains a negative cycle with $ n^3 / 2^{\Omega{(\sqrt{\log n})}} $ work and $ O (\log{n}) $ comparison depth.
\end{theorem}

\begin{proof}[Proof sketch]
First, note that the value of the minimum weight cycle in a directed graph can be found by computing $ \min_{e = (u, v) \in E} w (u, v) + \dist_G (v, u) $, i.e., the cycle of minimum weight among all cycles consisting of first an edge $ e = (u, v) $ and then the shortest path from $ u $ to $ v $ is the global minimum weight cycle.
If all pairwise distances are already given, then computing the value of the minimum weight cycle (and thus also checking for a negative cycle) can therefore be done with $ O (n^2) $ work and $ O (\log{n}) $ depth (again by a `tournament' approach).

The APSP problem can in turn be reduced to min-plus matrix multiplication~\cite{AhoHU76}.
Let $ M $ be the adjacency matrix of the graph where additionally all diagonal entries are set to $ 0 $.
Recall that the the all-pairs distance matrix is given $ M^{n-1} $, where matrix multiplication is performed in the min-plus semiring.
By repeated squaring, this matrix can be computed with $ O (\log{n}) $ min-plus matrix multiplications.
Williams's principal approach for computing the min-plus product $ C $ of two matrices $ A $ and $ B $ is as follows.
\begin{enumerate}
\item[(A1)] Split the matrices $ A $ and $ B $ into \emph{rectangular} submatrices of dimensions $ n \times d $ and $ d \times n $, respectively, where $ d = 2^{\Theta(\sqrt{\log{n}})} $, as follows: For every $ 1 \leq k \leq \lceil n/d \rceil - 1 $, $ A_k $ contains the $k$-th group of $ d $ consecutive columns of $ A $ and $ B_k $ contains the $k$-th group of $ d $ consecutive rows of $ B $; for $ k = \lceil n/d \rceil $, $ A_k $ contains the remaining columns of $ A $ and $ B_k $ contains the remaining rows of $ B $.
\item[(A2)] For each $ 1 \leq k \leq \lceil n/d \rceil $, compute $ C_k $, the min-plus product of $ A_k $ and $ B_k $ (using the algorithm described below).
\item[(A3)] Determine the min-plus product of $ A $ and $ B $ by taking the entrywise minimum $ C := \min_{1 \leq k \leq \lceil n/d \rceil} C_k $.
\end{enumerate}

To carry out Step~(A2), Williams first uses a preprocessing stage applied to each pair of matrices $ A_k $ and $ B_k $ (for $ 1 \leq k \leq \lceil n/d \rceil $) individually.
It consists of the following three steps:
\begin{enumerate}
\item[(B1)] Compute matrices $ A_k^* $ and $ B_k^* $ of dimensions $ n \times d $ and $ d \times n $, respectively, as follows: Set $ A_k^* [i, p] := A_k [i, p] \cdot (n + 1) + p $, for every $ 1 \leq i \leq n $ and $ 1 \leq p \leq d $, and set $ B_k^* [q, j] := B_k [q, j] \cdot (n + 1) + q $, for every $ 1 \leq q \leq d $ and $ 1 \leq j \leq n $.
\item[(B2)] Compute matrices $ A_k' $ and $ B_k' $ of dimensions $ n \times d^2 $ and $ d^2 \times n $, respectively, as follows: Set $ A'_k [i, (p, q)] := A_k^* [i, p] - A_k^* [i, q] $, for every $ 1 \leq i \leq n $ and $ 1 \leq p, q \leq d $, and set $ B_k' [(p, q), j] := B_k^* [q, j] - B_k^* [p, j] $, for every $ 1 \leq j \leq n $, $ 1 \leq p, q \leq d $.
\item[(B3)] For every pair $ p, q $ ($ 1 \leq p, q \leq d) $, compute and sort the set $ S_k^{p, q} := \{ A_k' [i, (p, q)] \mid 1 \leq i \leq n \} \cup \{ B_k' [(p, q), j] \mid 1 \leq j \leq n \}  $, where ties are broken such that entries of $ A_k' $ have precedence over entries of $ B_k' $.
Then compute matrices $ A_k'' $ and $ B_k'' $ of dimensions $ n \times d^2 $ and $ d^2 \times n $, respectively, as follows: Set $ A_k'' [i, (p, q)] $ to the \emph{rank} of the value $ A_k' [i, (p, q)] $ in the sorted order of $ S_k^{p, q} $, for every $ 1 \leq i \leq n $ and $ 1 \leq p, q \leq d $, and set $ B_k'' [(p, q), j] $ to the \emph{rank} of the value $ B'_k [(p, q), j] $ in the sorted order of $ S_k^{p, q} $, for every $ 1 \leq j \leq n $ and $ 1 \leq p, q \leq d $.
\end{enumerate}
This type of preprocessing is also known as Fredman's trick~\cite{Fredman76}.
As Williams shows, the problem of computing $ C_k $ now amounts to finding, for every $ 1 \leq i \leq n $ and $ 1 \leq j \leq n $, the unique $ p^* $ such that $ A''_{i, (p^*, q)} \leq B''_{(p^*, q), j} $ for all $ 1 \leq q \leq d $.
Using tools from circuit complexity and fast rectangular matrix multiplication, this can be done in time $ \tilde O (n^2) $, either with a randomized algorithm~\cite{Williams14}, or, with slightly worse constants in the choice of $ d $ (and thus the exponent of the overall algorithm), with a deterministic algorithm~\cite{ChanW16}.
The crucial observation for our application is that after the preprocessing stage no comparisons involving the input values are performed anymore since all computations are performed with regard to the matrices $ A_k'' $ and $ B_k'' $, which only contain the ranks (i.e., integer values from $ 1 $ to $ 2 n $).

The claimed work bound follows from Williams's running time analysis.
We can bound the comparison depth as follows.
First note that apart from Steps~(A3) and~(B3) we only incur $ O (\log n) $ overhead in the depth.
Step~(A3) can be implemented with $ O (\log{n}) $ depth by using a tournament approach for finding the respective minima.
For Step~(B3) we can use a parallel version of merge sort on $ n $ items that has work $ O (n \log{n}) $ and depth $ O (\log{n}) $~\cite{Cole88}. 
\end{proof}

We now apply the reduction of Theorem~\ref{thm:parametric search} to the algorithm of Theorem~\ref{thm:parallel min plus} to find $ \lambda^* $, the value of the minimum ratio cycle, in time $ n^3 / 2^{\Omega{(\sqrt{\log n})}} $ (using $ T_\text{s} (n) = W_\text{p} (n) = n^3 / 2^{\Omega{(\sqrt{\log n})}} $, and $ D_\text{p} (n) = O (\log{n}) $).
We then use Williams' APSP algorithm to find a minimum weight cycle in $ G_{\lambda^*} $ in time $ n^3 / 2^{\Omega{(\sqrt{\log n})}} $.
We thus arrive at the following guarantees of the overall algorithm.

\begin{corollary}\label{cor:min ratio dense graphs}
There is deterministic algorithm for computing a minimum ratio cycle with running time $ n^3 / 2^{\Omega{(\sqrt{\log n})}} $.
\end{corollary}

\section{Conclusion}
We have presented a faster strongly polynomial algorithm for finding a cycle of minimum cost-to-time ratio, a problem which has a long history in combinatorial optimization and recently became relevant in the context of quantitative verification.
Our approach combines parametric search with new parallelizable single-source shortest path algorithms and also yields small improvements for graphs of constant treewidth and in the dense regime.
The main open problem is to push the running time down to $ \tilde O (m n) $, nearly matching the strongly polynomial upper bound for the less general problem of finding a minimum mean cycle.

\printbibliography[heading=bibintoc] 

\end{document}